\documentclass[intlimits,twoside,a4paper]{article}
\usepackage{amsmath,amsthm}
\usepackage[cp1251]{inputenc}
\usepackage[eqsecnum]{cmpj3}

\newcommand{\bea}{\begin{eqnarray}}
\newcommand{\eea}{\end{eqnarray}}
\newcommand{\ba}{\begin{array}}
\newcommand{\ea}{\end{array}}
\newcommand{\edc}{\end{document}}
\newcommand{\bc}{\begin{center}}
\newcommand{\ec}{\end{center}}
\newcommand{\be}{\begin{equation}}
\newcommand{\ee}{\end{equation}}

\newtheorem{thm}{Theorem}[section]
\newtheorem{lem}[thm]{Lemma}

\newtheorem{prop}[thm]{Proposition}
\newtheorem{defin}[thm]{Definition}
\theoremstyle{remark}
\newtheorem{rem}{Remark}[section]

\issue{2021}{24}{1}{13001}
\doinumber{10.5488/CMP.24.13001}

\title[Determination of paramagnetic and ferromagnetic phases]
{Determination of paramagnetic and ferromagnetic phases of an
Ising model on a third-order Cayley tree}
\author[H. Ak\i n] {H. Ak\i n\footnote{email: akinhasan25@gmail.com}}
\address{Ceyhun Atuf Kansu Caddesi 1164. Sokak, 9/4, TR06105, \c{C}ankaya, Ankara,
Turkey}

\date{Received March 28, 2020, in final form September 23, 2020}
\begin{document}
\maketitle

\begin{abstract}
In this present paper,  the recurrence equations of an Ising model
with three coupling constants  on a third-order Cayley tree are
obtained. Paramagnetic and ferromagnetic phases associated with
the Ising model are characterized. Types of phases and partition
functions corresponding to the model are rigorously studied. Exact
solutions of the mentioned model are compared with the numerical
results given in Ganikhodjaev {et al.} [ J. Concr. Appl. Math., 2011, \textbf{9}, No.~1, 26--34].

\end{abstract}

\section{Introduction}

In magnetic and ferroelectric systems, phase diagrams of a model
with various transition lines and modulated phases are obtained
with the presence of different competing interactions
\cite{Inawashiro}.
In order to picture all phase diagrams of a given model, many
researchers  used the recurrence relations associated with the
model and observed their dynamical properties after a large number
of iterations
\cite{Vannimenus,NHSS,UGAT2012IJMPC,MTA1985a,Yokoi-Oliveira-S1985}.
Recently, in references \cite{Lebowitz1,Vannimenus,MTA1985a}, some
features of complex phase diagrams corresponding to the ANNNI
(Axial Next-Nearest-Neighbor Ising) model consisting of an Ising
spin Hamiltonian on a Cayley tree were studied.

The existence and quantification of the modulated phase diagrams
after  iterations of relevant recurrence equations as a way of
probing the ground phases of Ising model have gained much attention~\cite{GAUT2011Chaos,GU2011a,GTA,Vannimenus,NHSS1,AUT2010AIP,UA2011CJP,UA2010PhysicaA,UGAT2012ACTA,UGAT2012IJMPC,MTA1985a,Inawashiro,Inawashiro-T1983}.
In general, many authors have plotted the phase diagrams
associated with the model by means of the periodic fixed points of
an operator consisting of recurrent relations
\cite{GAUT2011Chaos,GU2011a,GTA,Vannimenus,NHSS1,AUT2010AIP,Moraal,UGAT2012IJMPC,MTA1985a}.
In the above-mentioned works, most results are numerically
obtained. In \cite{Chakraborty1992}, Chakraborty  investigated
the effects of inclusion of three-spin and four-spin couplings to
the Ising model. Furthermore, Chakraborty  employed the
molecular-field approximation to study the effects of both
three-spin and four-spin couplings on the field-free and
field-induced phase transitions possible for the model. In this
paper, we obtain our results in an analytical way by comparing the
numerical results.

Horiguchi \cite{Horiguchi1986}  proved that there exists a
phase transition in the Ising model on the square lattice with
two-spin interactions in the vertical direction and with
slantwise, alternate three-spin interactions in the horizontal
direction. Moreover, Horiguchi  proved that the system described by
Hamiltonian is in the paramagnetic state at high temperatures and
obtained an upper bound to the critical temperature. Azhari
{et al.} \cite{Azhari2017}  investigated the magnetic
properties of the mixed spin-1/2 and spin-1 Ising ferromagnetic
system with four-spin interaction $J_4$ and next-nearest neighbor
(NNN) coupling $J_0$. They are interested in the phase diagram and
in the location and the multitude of the compensation point.

In \cite{UGAT2012IJMPC}, we numerically studied the Lyapunov exponent
and modulated phases for the Ising model with different coupling
constants on an arbitrary-order Cayley tree. We also plotted the
variation of the wavevector $q$ with temperature in the modulated
phases. In \cite{Akin-Saygili2015} we  described the existence
of a phase transition problem by means of Gibbs measures of the
Potts model on an order three Cayley tree. In \cite{ART}, we 
constructed a class of new Gibbs measures by extending the known
Gibbs measures for the Ising model on a Cayley tree of the order $k_0$
to a Cayley tree of higher order $k > k_0$. Nazarov and Rozikov
 established the invariant subsets of an operator given in
\cite{Nazarov-Rozikov} and identified the periodic Gibbs measures
with period two by means of the operator. Here, in order to obtain
the periodic fixed points associated with the recurrence
equations, we use  similar methods in \cite{AGUT,AGTU2013ACTA}.

In \cite{Akin2017}, the author studied the Gibbs measures
associated with Vannimenus-Ising model for the compatible
conditions. He proved the existence of the
translation-invariant Gibbs measures with respect to the
compatible conditions.

In  \cite{Akin2016,Akin2017,Akin2017a,AT1}, we 
considered an external magnetic field to investigate the Gibbs
measures (see \cite{BRZ,BleherG}). In \cite{RAU}, we 
analytically studied the recurrence equations of an Ising model
with two coupling constants on a second-order Cayley tree without
considering the numerical investigation. Therefore, our present
results differ from the ones obtained in the mentioned papers.

The main purpose of the paper is to analytically derive the
recurrence equations of an Ising model with three coupling
constants on a third-order Cayley tree and to obtain the
paramagnetic, ferromagnetic and 2-period phases of the model by
means of the related recurrence equations by using the method
given in~\cite{RAU}. We exactly describe the paramagnetic
phase of the Ising model. We investigate some phase diagrams of an
Ising model with 2-spin couplings between the nearest neighbors
and the next-nearest neighbors, plus a 3-spin interaction, on a
Cayley tree with coordination number $z=4$. We obtain some
rigorous results: critical temperatures and curves, number of
phases, the partition functions. This model was numerically
examined by Ganikhodjaev {et al}. \cite{GAUT2011Chaos} on
semi-infinite second-order Cayley tree. 

\section{Preliminary}
\subsection{Cayley tree}

Any $k$ ($k>1$)-order Cayley tree $\Gamma^{k}$ is a weave pattern
in which $(k + 1)$ edges from each vertex point extend infinitely
as shown in figure~\ref{cayley-tree-k=3} ($k=3$). For the Cayley
tree shown as $\Gamma^{k}=(V,\Lambda)$, $V$ denotes the vertices
of the Cayley tree and $\Lambda$ denotes the set of edges. If
there is an edge $\ell$ joining two vertices $x$ and $y$, it is
called ``nearest neighbor'' and $\ell= < x, y > $. The distance of $x$
and $y$, over $V$ is defined as $d(x, y)$, the shortest path
between $x$ and $y$. 
The set of edge points in $V_n$ is represented as $L_n$ (see
\cite{Akin2017} for details). Note that we  consider a
semi-infinite third-order Cayley tree which has got uniformly bounded
degrees.
\begin{figure} [!t]
\centering
\includegraphics[width=60mm]{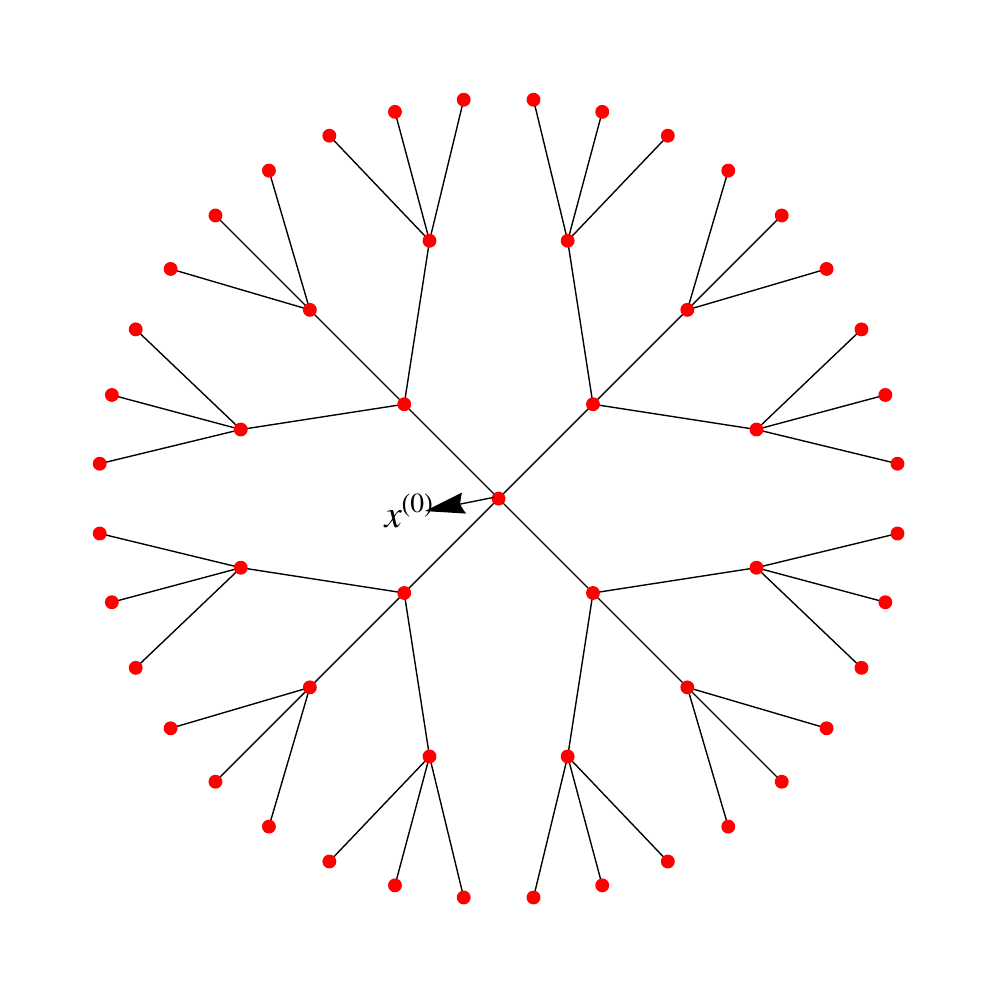}
\caption{(Colour online) Three-order Cayley tree with three
levels.}\label{cayley-tree-k=3}
\end{figure}
The set of all vertices with the distance $n$ from the root  $x^{(0)}$
is called the $n$th level of $\Gamma^{k}$ and we denote the sphere
of radius $n$ on $V$ by
$$ W_n=\{x\in V:
d(x,x^{(0)})=n \}
$$ and the ball of radius $n$ by
$$ V_n=\{x\in V:
d(x,x^{(0)})\leq n \}.
$$ The set of direct successors of any vertex
$x\in W_n$ is denoted by
$$
S(x)=\{y\in W_{n+1}:d(x,y)=1 \}.
$$
Hereafter, we  use the following definitions for
neighborhoods.
\begin{defin}\label{neighborhoods}
\begin{enumerate}
\item Two vertices $x$ and $y$, $x,y \in V$ are called {\it
\textbf{nearest-neighbors (NN)}} if there exists an edge
$\ell\in\Lambda$ connecting them, which is denoted by
$\ell=< x, y >$. \item The next-nearest-neighbor vertices $x\in W_n$
and $z\in W_{n+2}$ ($z\in S^2(x)$) are called {\it
\textbf{prolonged next-nearest-neighbors (PNNN)}} if $|x|\neq |z|$
and is denoted by $>x,z<$ \item The triple of vertices $x,y,z$ is
called {\it \textbf{ternary prolonged next-nearest-neighbors}} if
$x\in W_n,y\in S(x)$ and $z\in S(y)$ ($x\in W_n,y\in W_{n+1}$ and
$z\in W_{n+2}$) for some nonnegative integer $n$  and is denoted
by $>x,y,z<$.
\end{enumerate}
\end{defin}
In this paper, we  consider a Hamiltonian with
\textbf{competing nearest-neighbor interactions},
\textbf{prolonged next-nearest-neighbors (PNNN)}  and
\textbf{ternary prolonged next nearest-neighbor interactions}.
Therefore, we can state the Hamiltonian by
\begin{eqnarray} H(\sigma )&=&-J\sum _{
\begin{array}{l}
 < x, y > \\
\end{array}
} \sigma (x)\sigma (y)-J_p\sum _{
\begin{array}{l}
 >x,z< \\
\end{array}
} \sigma (x)\sigma (z)\\\nonumber &&-J_t\sum _{
\begin{array}{l}
 >x,y,z< \\
\end{array}
} \sigma (x)\sigma (y)\sigma (z),
\label{Ham-Ist}
\end{eqnarray}
where $J,J_p, J_t\in \mathbb{R}$ are coupling constants and $<x,
y>$ stands for NN vertices, $>x,z<$ stands for prolonged NNN and
$>x,y,z<$ stands for prolonged ternary NNN.
\section{Definitions and equations}
In this section, we  construct the recurrent equations
associated with the Hamiltonian \eqref{Ham-Ist}.
It is well-known that there are different approaches for
determining the Gibbs measures (phases) of lattice models such as
Ising and Potts with competing interactions in the literature
\cite{GATTMJ,AT1,Rozikov,Bleher-Zalys,Bleher1990a}. Recently, some
researchers~\cite{Akin2017a,AT1,Gandolfo2012,Kindermann} 
determined the Gibbs measures associated with the Ising models by
means of the fixed points of operators obtained by the partition
functions. In our work, we  use the method based on the
recursive equations to derive the phases.

After specifying a Hamiltonian $H$, the equilibrium state of a
physical system with Hamiltonian $H$ is described by the
probability measure
$$
\mu(\sigma _n)=\frac{\exp [-\beta H(\sigma _n)]}{{\underset{\eta
_n\in \{-1,+1\}^{V_n}}{\sum }}\exp [-\beta H(\eta_n)]},
$$
where $\beta$ is a positive number which is proportional to the
inverse of the absolute temperature. The above $\mu$ is called the
Gibbs distribution relative to $H$. The standard approach consists
in writing down recurrence equations relating the partition
function
$$
Z_n=\sum _{\eta _n\in \{-1,+1\}^{V_n}} \exp [-\beta H(\eta_n)],
$$
of an $n$-generation tree to the partition function $Z_{n-1}$ of
its subsystems containing $(n-1)$ generations [see figure~\ref{fig2}].

Many authors  introduced the notions of ground states (Gibbs
measure) of the Ising model with competing interactions on the
Cayley tree \cite{Akin2016,Akin2017,MAKfree2017,AkinRG2020}. Here,
we consider a shorter notation to write down the recurrence system
\begin{eqnarray*}\label{rec-eq2}
\left\{
\begin{array}{l}
 z_1=Z_n\left(
\begin{array}{ccc}
 + & + & + \\
  & + &
\end{array}
\right), \\
 z_2=Z_n\left(
\begin{array}{ccc}
 + & + & - \\
   & + &
\end{array}
\right)=Z_n\left(
\begin{array}{ccc}
 + & - & + \\
   & + &
\end{array}
\right)=Z_n\left(
\begin{array}{ccc}
 - & + & + \\
   & + &
\end{array}
\right), \\
 z_3=Z_n\left(
\begin{array}{ccc}
 + & - & - \\
  & + &
\end{array}
\right)=Z_n\left(
\begin{array}{ccc}
 - & + & - \\
   & + &
\end{array}
\right)=Z_n\left(
\begin{array}{ccc}
 - & - & + \\
 & + &
\end{array}
\right), \\
 z_4=Z_n\left(
\begin{array}{ccc}
 - & - & - \\
  & + &
\end{array}
\right),
\end{array}
\right.
\end{eqnarray*}
\begin{eqnarray*}\label{rec-eq2}
\left\{
\begin{array}{l}
 z_5=Z_n\left(
\begin{array}{ccc}
 + & + & + \\
  & - &
\end{array}
\right), \\
 z_6=Z_n\left(
\begin{array}{ccc}
 + & + & - \\
   & - &
\end{array}
\right)=Z_n\left(
\begin{array}{ccc}
 + & - & + \\
   & - &
\end{array}
\right)=Z_n\left(
\begin{array}{ccc}
 - & + & + \\
   & - &
\end{array}
\right),\\
 z_7=Z_n\left(
\begin{array}{ccc}
 + & - & - \\
  & - &
\end{array}
\right)=Z_n\left(
\begin{array}{ccc}
 - & + & - \\
   & - &
\end{array}
\right)=Z_n\left(
\begin{array}{ccc}
 - & - & + \\
 & - &
\end{array}
\right). \\
 z_8=Z_n\left(
\begin{array}{ccc}
 - & - & - \\
  & - &
\end{array}
\right).
\end{array}
\right.
\end{eqnarray*}
\begin{figure} [!t]
\centering
\includegraphics[width=80mm]{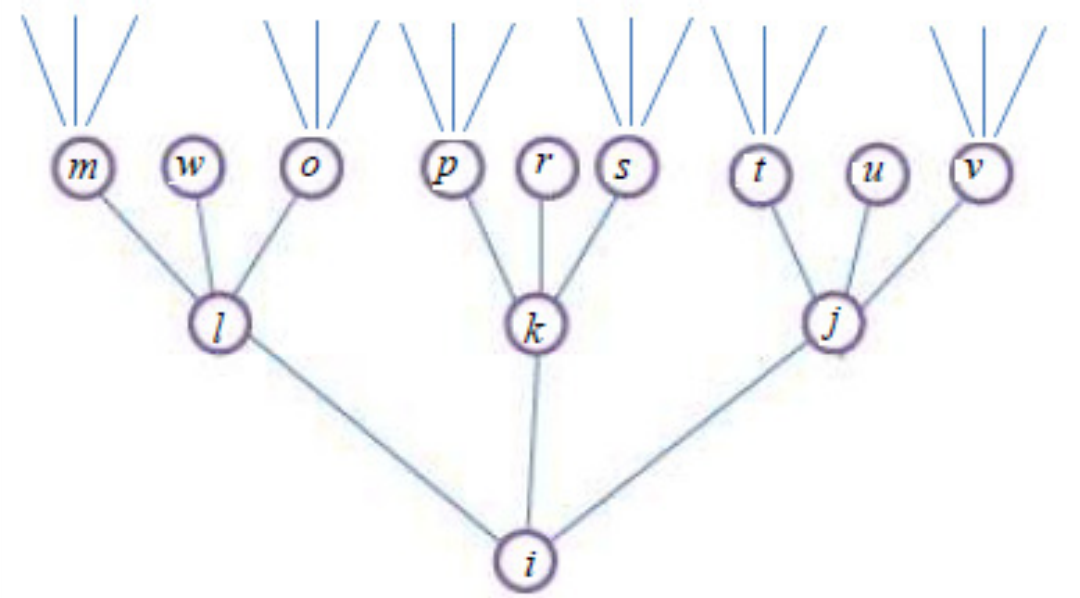}
\caption{(Colour online) Configurations on semi-finite Cayley tree of the order three
with levels 2. Schematic diagram  illustrating the summation
used in equation \eqref{partition1aa}.}\label{fig2}
\end{figure}
For the sake of simplicity, we denote $Z_{(n+1)}\left(
\begin{array}{ccc}
 l & k & j \\
  & i &
\end{array}
\right)=Z_{(n+1)}\left(i;l,k,j\right)$. We have
\begin{eqnarray}\label{partition1aa}
Z_{n+1}\left(i;l, k, j\right)&=&\sum _{m,w,o,p,r,s,t,u,v\in
\{-1,+1\}} [\exp
(A(m,w,o,p,r,s,t,u,v))\\
&&\times Z_{n}\left(l;m,w,o\right)Z_{n}\left(k;p,r,s\right)Z_{n}\left(j;t,u,v\right)]\,,\nonumber
\end{eqnarray}
where
\begin{eqnarray*}
A(m,w,o,p,r,s,t,u,v)&=&Ji(j+k+l)+J_pi(m+w+o+p+r+s+t+u+v)\\&&+J_t\left\lbrace i[l(m+w+o)+k(p+r+s)+j(t+u+v)]\right\rbrace ,
\end{eqnarray*}
$i,j,k,l\in \{-1,+1\}$ and $n=1,2,\ldots $ (see figure
\ref{fig2}).

In physics, it is well known that a partition function describes
the statistical properties of a system in thermodynamic
equilibrium \cite{Preston,Kindermann}. Therefore, we should derive
the partition functions associated with the Hamiltonian
\eqref{Ham-Ist}. From the equation \eqref{partition1aa} by using
the iterative method, we can derive the partial partition
functions as follows:
$$
\left\{
\begin{array}{l}
 z_1^{(n+1)}=a^3(bc)^{-9}\left((bc)^6z_1^{(n)}+3(bc)^4z_2^{(n)}+3(bc)^2z_3^{(n)}+z_4^{(n)}\right)^3 \\
 z_2^{(n+1)}=a(bc)^{-9}\left((bc)^6z_1^{(n)}+3(bc)^4z_2^{(n)}+3(bc)^2z_3^{(n)}+z_4^{(n)}\right)^2 \\
 \ \ \ \ \ \ \  \ \ \text{}\times \left(b^6z_5^{(n)}+3c^2b^4z_6^{(n)}+3b^2c^4z_7^{(n)}+c^6z_8^{(n)}\right) \\
 z_3^{(n+1)}=a^{-1}(bc)^{-9}\left((bc)^6z_1^{(n)}+3(bc)^4z_2^{(n)}+3(bc)^2z_3^{(n)}+z_4^{(n)}\right) \\
 \ \ \ \ \ \ \  \ \ \text{}\times \left(b^6z_5^{(n)}+3c^2b^4z_6^{(n)}+3b^2c^4z_7^{(n)}+c^6z_8^{(n)}\right)^2 \\
 z_4^{(n+1)}=a^{-3}(bc)^{-9}\left(b^6z_5^{(n)}+3c^2b^4z_6^{(n)}+3b^2c^4z_7^{(n)}+c^6z_8^{(n)}\right)^3
\end{array}
\right.
$$
$$
\left\{
\begin{array}{l}
 z_5^{(n+1)}=a^{-3}(bc)^{-9}\left(z_1^{(n)}+3(bc)^2z_2^{(n)}+3(bc)^4z_3^{(n)}+(bc)^6z_4^{(n)}\right)^3 \\
 z_6^{(n+1)}=a^{-1}(bc)^{-9}\left(z_1^{(n)}+3(bc)^2z_2^{(n)}+3(bc)^4z_3^{(n)}+(bc)^6z_4^{(n)}\right)^2 \\
  \ \ \ \ \ \ \  \ \ \text{}\times \left(c^6z_5^{(n)}+3c^4b^2z_6^{(n)}+3c^2b^4z_7^{(n)}+b^6z_8^{(n)}\right) \\
 z_7^{(n+1)}=a(bc)^{-9}\left(z_1^{(n)}+3(bc)^2z_2^{(n)}+3(bc)^4z_3^{(n)}+(bc)^6z_4^{(n)}\right) \\
  \ \ \ \ \ \ \  \ \ \text{}\times \left(c^6z_5^{(n)}+3c^4b^2z_6^{(n)}+3c^2b^4z_7^{(n)}+b^6z_8^{(n)}\right)^2 \\
 z_8^{(n+1)}=a^3(bc)^{-9}\left(c^6z_5^{(n)}+3c^4b^2z_6^{(n)}+3c^2b^4z_7^{(n)}+b^6z_8^{(n)}\right)^3,
\end{array}
\right.
$$
where $a=\re^{\beta J}$, $b=\re^{\beta J_p}$, $c=\re^{\beta J_t}$.
Noting that
\begin{eqnarray*}
&&\left(z_2^{(n+1)}\right)^3=\left(z_1^{(n+1)}\right)^2z_4^{(n+1)},\
\left(z_3^{(n+1)}\right)^3=z_1^{(n+1)}\left(z_4^{(n+1)}\right)^2,\\
&&\left(z_6^{(n+1)}\right)^3=\left(z_5^{(n+1)}\right)^2z_8^{(n+1)},\
\left(z_7^{(n+1)}\right)^3=z_5^{(n+1)}\left(z_8^{(n+1)}\right)^2,
\end{eqnarray*}
we obtain independent variables, introducing the new variables
$u_i^{(n+1)}=(z_i^{(n+1)})^{\frac{1}{3}}$, and we get the recurrence
system of the following simpler form:
\begin{eqnarray}\label{rec-eq2}
\left\{
\begin{array}{l}
 u_1^{(n+1)}=\frac{a}{(bc)^3}\left((bc)^2u_1^{(n)}+u_4^{(n)}\right)^3 \\
 u_4^{(n+1)}=\frac{1}{a(bc)^3}\left(b^2u_5^{(n)}+c^2u_8^{(n)}\right)^3 \\
 u_5^{(n+1)}=\frac{1}{a(bc)^3}\left(u_1^{(n)}+(bc)^2u_4^{(n)}\right)^3 \\
 u_8^{(n+1)}=\frac{a}{(bc)^3}\left(c^2u_5^{(n)}+b^2u_8^{(n)}\right)^3.
\end{array}
\right.
\end{eqnarray}
Let us define the operator as
$$
F: u^{(n)}=(u_1^{(n)}, u_4^{(n)}, u_5^{(n)}, u_8^{(n)})\in
\mathbb{R}^{4}_+ \rightarrow F(u^{(n)}) = (u_1^{(n+1)},
u_4^{(n+1)}, u_5^{(n+1)}, u_8^{(n+1)})\in \mathbb{R}^{4}_+.
$$
Then, we can write the recurrence equations \eqref{rec-eq2} as
$u^{(n+1)} = F(u^{(n)}), n > 0$ which in the theory of dynamical
systems is called a trajectory of the initial point $u^{(0)}$
under the action of the operator $F$. In this way we can specify
the asymptotic behavior of the partition functions $Z_n$ for
$n\rightarrow \infty$ by means of the trajectory of $u^{(0)}$
under the action of the operator $F$. We examine the
dynamical system for a given initial point $u^{(0)}\in
\mathbb{R}^{4}_+$ in detail.
\section{Dynamics behavior of the operator $F$}\label{Dynamics of the
operator}

In this section we identify the fixed points of the operator
$F$ given in~\eqref{rec-eq2} by using the method given in
\cite{RAU} and \cite[chapter 3]{Rozikov}. Here, we assume that
$c=1$, or equivalently $J_t=0$.

Denote the set of the fixed points of the operator $F$ by
$$
\text{Fix}(F) = \{u : F(u) = u\}.
$$
From \eqref{rec-eq2}, by substituting the new variables $\alpha^3
=a$, $(v_i^{(n)})^3 = u_i^{(n)}$, for $i = 1, 4, 5, 8$. Then, one
gets a new operator as follows:
\begin{eqnarray}\label{rec-eq3a}
\left\{
\begin{array}{l}
 v_1^{(n+1)}=\frac{\alpha}{(bc)}\left((bc)^2(v_1^{(n)})^{3}+(v_4^{(n)})^{3}\right) \\
 v_4^{(n+1)}=\frac{1}{\alpha (bc)}\left(b^2(v_5^{(n)})^{3}+c^2(v_8^{(n)})^{3}\right) \\
 v_5^{(n+1)}=\frac{1}{\alpha(bc)}\left((v_1^{(n)})^{3}+(bc)^2(v_4^{(n)})^{3}\right)\\
 v_8^{(n+1)}=\frac{\alpha}{(bc)}\left(c^2(v_5^{(n)})^{3}+b^2(v_8^{(n)})^{3}\right).
\end{array}
\right.
\end{eqnarray}
We  deal with the fixed points of a system of recurrent
equations given in \eqref{rec-eq3a}. To this end, it suffices to
solve the system of equations

\begin{eqnarray}\label{rec-eq3}
\left\{
\begin{array}{l}
 v_1=\frac{\alpha}{(bc)}\left((bc)^2(v_1)^{3}+(v_4)^{3}\right) \\
 v_4=\frac{1}{\alpha (bc)}\left(b^2(v_5)^{3}+c^2(v_8)^{3}\right) \\
 v_5=\frac{1}{\alpha(bc)}\left((v_1)^{3}+(bc)^2(v_4)^{3}\right)\\
 v_8=\frac{\alpha}{(bc)}\left(c^2(v_5)^{3}+b^2(v_8)^{3}\right).
\end{array}
\right.
\end{eqnarray}
Let us consider the following set:
\begin{eqnarray}\label{set-invariant}
A=\{(v_1, v_4, v_5, v_8)\in
\mathbb{R}^{4}_+:v_1=v_8,v_4=v_5,c=1\}.
\end{eqnarray}
Note that the set $A$ is invariant with respect to $F$ i.e.,
$F(A)\subset A.$
\begin{lem}\label{Lemma-fixed-set}
If a vector $\textbf{u}$ is a fixed point of the operator $F$,
then $\textbf{u}\in M_1:=\{\textbf{u}=(u_1,u_4,u_5,u_8)\in
\mathbf{R}^{4}_{+}: u_1=u_8,u_4=u_5\}$ or $\textbf{u}\in
M_2:=\{\textbf{u}=(u_1,u_4,u_5,u_8)\in \mathbf{R}^{4}_{+}:
\left(\sqrt[3]{u_4}+\sqrt[3]{u_5}\right)^2-\sqrt[3]{u_4u_5}=\psi
\left(\left(\sqrt[3]{u_1}+\sqrt[3]{u_8}\right){}^2-\sqrt[3]{u_1u_8}\right)\}$,
where $\psi(y)
=\frac{b^3+\alpha  y (1-b^4)}{b^2 \alpha \left(b \alpha
y-1\right)}$.
\end{lem}
\begin{proof} From the system \eqref{rec-eq3}, we have
\begin{eqnarray}\label{fixed-1a}
&& \left(v_1-v_8\right)\left(b\alpha
\left(v_1^2+v_1v_8+v_8^2\right)-1\right)+\alpha
b^{-1}\left(v_4-v_5\right)\left(v_4^2+v_4v_5+v_5^2\right)=0,
\\\label{fixed-1b}
&& (\alpha
b)^{-1}(\left(v_1-v_8\right)(v_1^2+v_1v_8+v_8^2)+(v_4-v_5)(b^2v_4^2+b^2v_4v_5+b^2v_5^2+\alpha
b))=0.
\end{eqnarray}
From \eqref{fixed-1a} and \eqref{fixed-1b} one can conclude that
if $v_1=v_8$ (respectively $v_4=v_5$), then, $v_4=v_5$
(respectively $v_1=v_8$). Therefore, $v_4=v_5$ if and only if
$v_1=v_8$.

Now, let us assume that $v_1\neq v_8$ and $v_4\neq v_5$, then we
can reduce the equations \eqref{fixed-1a} and \eqref{fixed-1b} to
the following equation:
$$
\frac{\alpha  b^{-1}\left(v_4^2+v_4v_5+v_5^2\right)}{\left(b
\alpha
\left(v_1^2+v_1v_8+v_8^2\right)-1\right)}=\frac{\left(b^2v_4^2+b^2v_4v_5+b^2v_5^2+\alpha
b\right)}{\left(v_1^2+v_1v_8+v_8^2\right)}.
$$
Therefore, from the last equation we have
\begin{equation}\label{fixed-inv2}
v_4^2+v_4v_5+v_5^2=\frac{b^3+\alpha  \left(v_1^2+v_1
v_8+v_8^2\right)(1-b^4)}{b^2 \alpha  \left(b \alpha
\left(v_1^2+v_1 v_8+v_8^2\right)-1\right)}.
\end{equation}
The equation \eqref{fixed-inv2} gives
$\left(\sqrt[3]{u_4}+\sqrt[3]{u_5}\right)^2-\sqrt{3}{u_4u_5}=\psi
\left[\left(\sqrt[3]{u_1}+\sqrt[3]{u_8}\right)^2-\sqrt[3]{u_1u_8}\right]$
(see \cite{GU2011a}).
\end{proof}
\begin{rem}
The fixed points of the operator $F$ belonging to the set $M_2$
give the ferromagnetic phases corresponding to the Ising model
\eqref{Ham-Ist}. It is analytically very difficult to examine the fixed points of the operator $F$
belonging to the set $M_2$. The
ferromagnetic phase regions corresponding to the Ising model
\eqref{Ham-Ist} can be numerically  determined.
\end{rem}
\subsection{The existence of paramagnetic and ferromagnetic phases}\label{The paramagnetic phases}
In this subsection, we  analytically prove the existence of
paramagnetic and ferromagnetic phases for the Ising model
\eqref{Ham-Ist}. Therefore, we deal with the fixed points of the
operator $F$ belonging to the set $A$ given in the equation
\eqref{set-invariant}. 

Assume that $u_1^{(n)}=u_8^{(n)}, u_4^{(n)}= u_5^{(n)}$. After
replacing $x_n=\frac{v_1^{(n)}}{v_4^{(n)}}$, we get the following
system of recurrent equation
\begin{eqnarray}\label{c=1-fixed points-a}
x_{n+1}=\alpha^2\left(\frac{1+b^2(x_{n})^3}{b^2+(x_{n})^3}\right),
\end{eqnarray}
where $b>0, x_n>0$ and $\alpha>0$.

We will obtain the fixed points of the recurrent equation
\eqref{c=1-fixed points-a}. To this end, it suffices to solve the
following equation
\begin{eqnarray}\label{c=1-fixed points}
\alpha^{-2}x=f(x):=\left(\frac{1+b^2x^3}{b^2+x^3}\right),
\end{eqnarray}
Note that  in order to describe the phases (or the limiting
Gibbs measures) corresponding to the model, we take into account
the approach based on recurrent equations for partition functions
(see \cite{RAU,BleherG,AT1,GTA-CUBO2005,GR-2019} for details). We
denote the set of the Gibbs measures corresponding to the
Hamiltonian \eqref{Ham-Ist} by $\mathcal{G}_H$.

One can show that $f$ is bounded and thus the curve
$y=\left(\frac{1+b^2x^3}{b^2+x^3}\right)$ must intersect the line
$y = \alpha^{-2}x$. Therefore, our construction gives an element of
$\mathcal{G}_H$. Note that if the equation \eqref{c=1-fixed
points} has more than one solution, then our construction gives
more than one element of $\mathcal{G}_H$ (see \cite[proposition
10.7]{Preston} for details).

\begin{prop}\label{proposition1}
The equation \eqref{c=1-fixed points} 
has a unique solution if $0<b<1$.  Assume that $b>\sqrt{2}$, then
the equation \eqref{c=1-fixed points} has 2 solutions if either
$\eta_1(b)=\alpha^{-2}$ or $\eta_2(b)=\alpha^{-2}$. If
$\eta_1(b)<\alpha^{-2}<\eta_2(b)$, then 
the equation \eqref{c=1-fixed points} has 3 solutions. In fact, we
have
\begin{eqnarray*}
\eta_i
(b)&=&\frac{1}{x_i}\left(\frac{1+b^2x_i^3}{b^2+x_i^3}\right),
\end{eqnarray*}
where $x_i$ are the solutions of the equation
$b^2x^6-2\left(b^4-2\right)x^3+b^2=0$.
\end{prop}
\begin{proof}
Let us take the first and the second derivatives of the function
$f$, and we have
\begin{eqnarray}\label{first-derivative}
f'(x)=\frac{3\left(b^4-1\right) x^2}{\left(b^2+x^3\right)^2},
\end{eqnarray}
$$
f''(x)=\frac{6 \left(b^4-1\right) x \left(b^2-2 x^3\right)
}{\left(b^2+x^3\right)^3}.
$$
From \eqref{first-derivative}, assume that $b<1$ (with $x \geq 0$),
then $f$  decreases and there exists a single solution of the
equation $f(x)=\alpha^{-2}x.$ Thus, in order to examine the phase
transition of the model, we should solve the equation
$f(x)=\alpha^{-2}x$ for $b>1.$ One can show that the graph of
$y=f(x)$ over interval $(0,\sqrt[3]{\frac{b^2}{2}})$ is concave up,
and the graph of $y=f(x)$ over the interval
$(\sqrt[3]{\frac{b^2}{2}},\infty)$ is concave down. As a result,
there are at most 3 positive solutions for $f(x)=\alpha^{-2}x.$

According to Preston \cite[proposition 10.7]{Preston},
the rational function $f$ has more than one fixed point if and
only if equation $xf'(x) = f(x)$ has more than one solution, which
is the same as
$$
b^2 x^6-2\left(b^4-2\right)x^3+b^2=0.
$$
These roots are
$$
x_1=\sqrt[3]{\frac{-2+b^4-\sqrt{4-5
b^4+b^8}}{b^2}},x_2=\sqrt[3]{\frac{-2+b^4+\sqrt{4-5
b^4+b^8}}{b^2}}.
$$
Note that if $4-5 b^4+b^8\geqslant 0$, then the roots $x_1$ and $x_2$
are real numbers. In this case, $4-5 b^4+b^8\geqslant 0$ if and only if
$b\in [0,1]\cap[\sqrt{2},\infty )$.
\end{proof}

\subsubsection{An illustrative example}

Let us consider the following equation:
\begin{eqnarray}\label{c=1-fixed points12}
x=g(x):=\alpha^{2}\left(\frac{1+b^2x^3}{b^2+x^3}\right).
\end{eqnarray}

From \eqref{c=1-fixed points12}, one gets
\begin{equation}\label{polinomial4d}
P_4(x)=x^4-\alpha^2b^2x^3+b^2x-\alpha^2=0.
\end{equation}
To solve the equation \eqref{polinomial4d} analytically is a
rigorous and complicated problem. Therefore, we have manipulated
the polynomial equation via Mathematica \cite{Wolfram}. Here, we
 only deal with positive fixed points due to the
positivity of exponential functions.
\begin{figure}[!t]
\begin{minipage}[h]{0.49\linewidth}
\center{\includegraphics[width=1\linewidth]{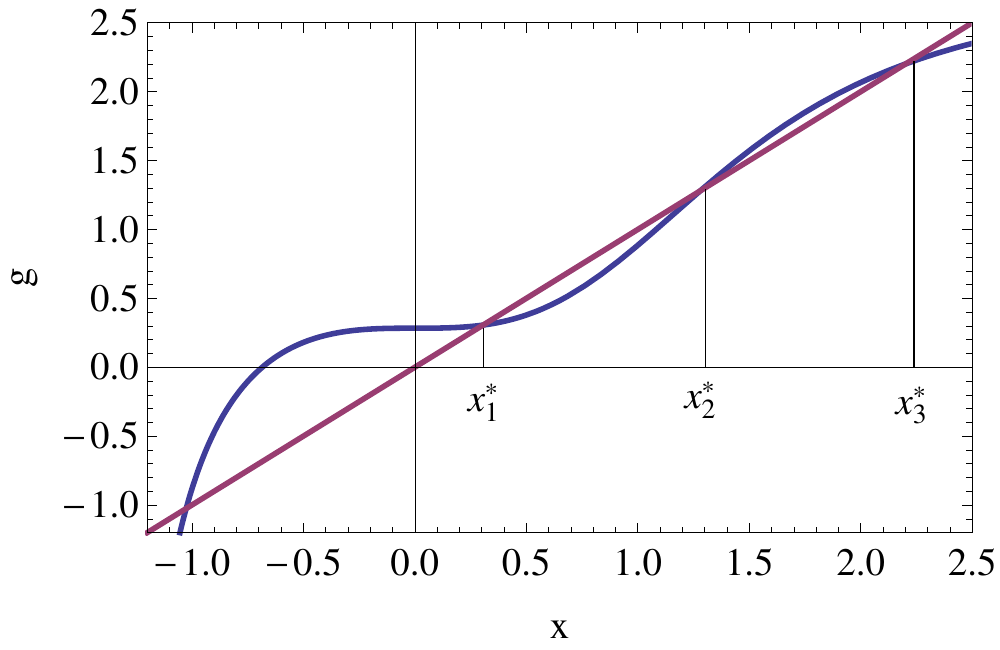}} (a) \\
\end{minipage}
\hfill
\begin{minipage}[h]{0.47\linewidth}
\center{\includegraphics[width=1\linewidth]{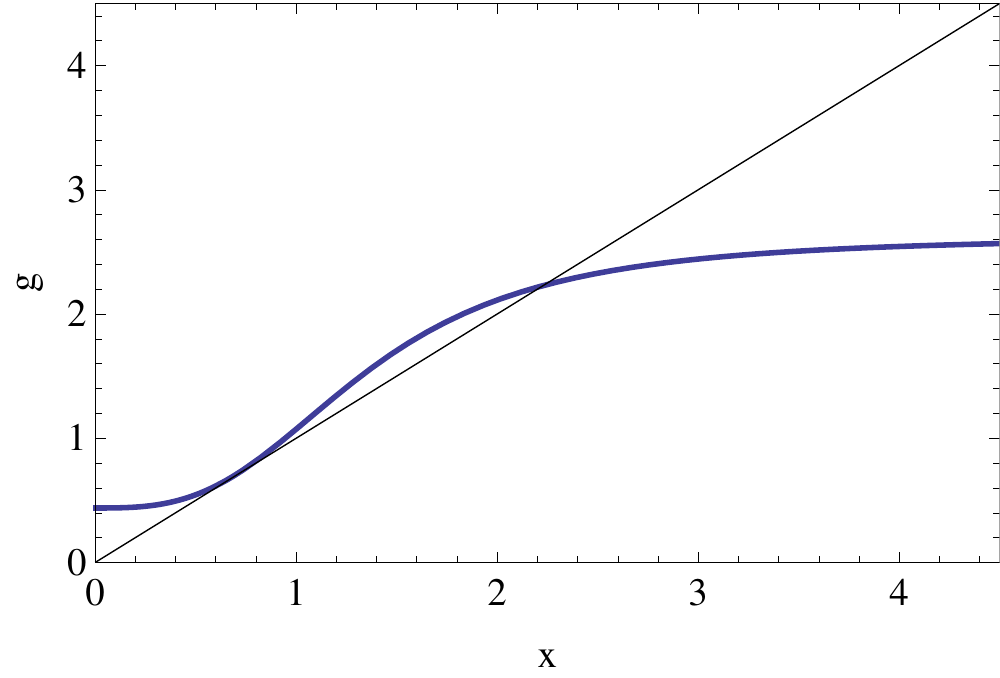}} \\(b)
\end{minipage}
\centering
\begin{minipage}[h]{0.47\linewidth}
\center{\includegraphics[width=1\linewidth]{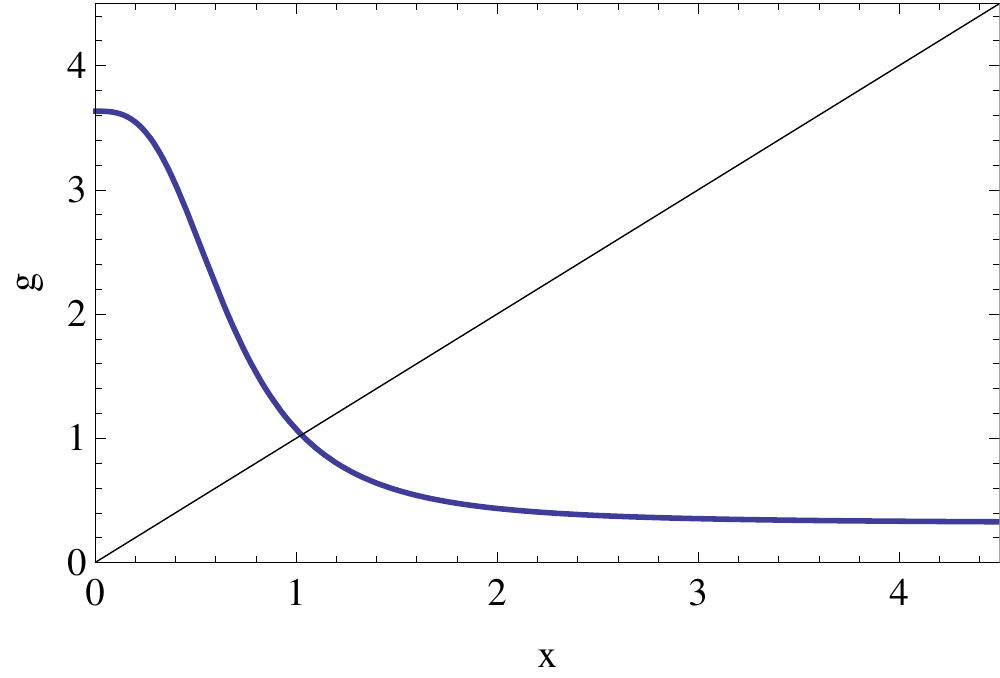}} (c) \\
\end{minipage}
\caption{(Colour online) (a) There exist three positive roots of
the equation \eqref{c=1-fixed points12} for $J = -4.5, J_p = 14, T
= 24.6$. (b) There exist two positive roots of the equation
\eqref{c=1-fixed points12} for $J = 2.6, T = 24.6, J_p= 11$. (c)
There exists only one positive root of the equation
\eqref{c=1-fixed points12} for $J = 2.6, T = 24.6, J_p=
-15$.}\label{3parametric-phase-c=1}
\end{figure}

We have obtained at most 3 positive real roots for some parameters
$J$, $J_p$ and $J_{t}=0$ (coupling constants) and temperature $T$.
As seen in figure~\ref{3parametric-phase-c=1} (a), the equation
\eqref{c=1-fixed points} has 3 positive fixed roots for $J = -4.5,
J_p = 14, J_{t}=0, T = 24.6$
($b=\re^{\frac{14}{24.6}}=1.77>\sqrt{2})$. These all fixed points
are $x_0^*= -1.02554,x_1^*=0.306205, x_2^*=1.28008,
x_3^*=2.20209$, respectively.

One can examine that $g'(0.306205)=0.219<1$ and
$g'(2.202)=0.59115<1$. Therefore, the fixed points $x_1^*=0.306205,
x_3^*=2.20209$ are stable (attracting) fixed points. Moreover, due to
$g'(1.28008)=1.3967>1$, the point $x_2^*=1.28008$ is an unstable
(repelling) fixed point.
The point $x_{cr}=1.15992$ is a breaking point of  the function
$g$. Thus, there are two extreme paramagnetic phases associated
with the positive fixed points.

There exist two positive roots of the equation \eqref{c=1-fixed
points12} for $J = 2.6, T = 24.6, J_p= 11$ [figure
\ref{3parametric-phase-c=1} (b)]. There exists only one positive
root for the equation \eqref{c=1-fixed points12} for $J = 2.6, T =
24.6, J_p= -15$ ($b=\re^{\frac{-15}{24.6}}=0.54<1$) [figure
\ref{3parametric-phase-c=1} (c)].

\subsection{The periodic points of the operator $F$}
One of the most interesting problems in non-linear dynamic systems
is to investigate the existence of periodic points. \cite{GR}.
In statistical physics, these periodic points reveal the phase
types corresponding to the given model
\cite{RAU,AT1,GTA-CUBO2005}.

We recall some definitions and results first.
\begin{defin}
A point $\textbf{u}=(u_1,u_4,u_5,u_8)$ in $\mathbf{R}^{4}_+$ is
called a periodic point of $F$ if there exists $p$ so that
$F^{p}(\textbf{u})=\textbf{u}$ where $F^{p}$ is the $p$th iterate
of $F$. The smallest positive integer $p$ satisfying the above is
called the prime period or the least period of the point $\textbf{u}$.
Denote by Per$_p(F)$ the set of periodic points with prime period
$p$.
\end{defin}
In order to find the periodic points of the operator $F$ with
$p=2$ on $M_1$, we  solve the equation $F(F(\textbf{u})) =
\textbf{u}$. Therefore, we deal with the solutions of equation
\begin{equation}\label{2-period1}
g(g(x))=x.
\end{equation}
Now, we consider the equation
$$
\frac{g(g(x))-x}{g(x)-x}=0.
$$
After simple calculations, we get
\begin{eqnarray}\label{polinomial6d}
&&p_6(x):=4b^2(1+b^4 \alpha ^6)x^6+ \alpha ^2(b^4-1)x^5+b^2 \alpha
^4(b^4-1)x^4 \\\nonumber &&+ 2 b^4(1+ \alpha^6)x^3+ b^2 \alpha
^2(b^4-1)x^2+\alpha ^4(b^4-1)x+b^2(b^4+\alpha ^6)=0.
\end{eqnarray}
In order to describe the periodic points with $p=2$ on $M_1$ of
the operator $F$, we should find the solutions to
\eqref{polinomial6d} which are different from the solutions of the
equation \eqref{polinomial4d}. In other words, we obtain the
set
\begin{equation}\label{2-period-set}
M_3:=\{(u_1,u_4,u_5,u_8)\in \mathbf{R}^{4}_+:g (g(x)) = x
\}.
\end{equation}
Therefore, we should examine the roots of the polynomial $p_6(x)$
of degree 6. As mentioned above, the roots of such polynomials can
be described using the known formulae. 
In order to illustrate the problem, we have manipulated the
equation \eqref{2-period1} via Mathematica \cite{Wolfram} [see
figure \ref{periodic2-phase} (red color)]. The black graph in
figure \ref{periodic2-phase} represents the roots of the nonlinear
function $y=g(x).$
\begin{figure} [!t]
\centering
\includegraphics[width=70mm]{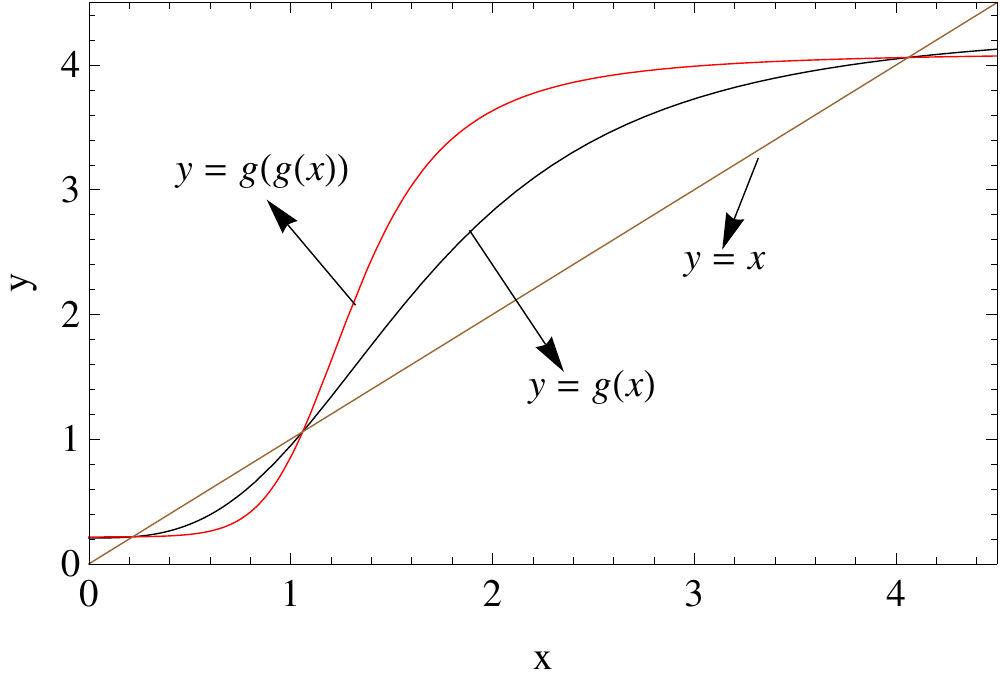}
\caption{(Colour online) There exist three positive roots of the
equation \eqref{c=1-fixed points} (black color) for $J = -6$,  $J_p =
56.9$,  $T = 75$. Moreover, the equation \eqref{2-period1} has three
positive roots (red color) for $J = -6$,  $J_p = 56.9$,  $T =
75$.}\label{periodic2-phase}
\end{figure}

We can obtain an initial point of the sequence
$(u_1^{(n)},u_4^{(n)},u_5^{(n)},u_8^{(n)})$ under positive
boundary condition as follows:
$$
\textbf{u}_0=(u_1^{(0)},u_4^{(0)},u_5^{(0)},u_8^{(0)})=(ab^3,\frac{b^3}{a},\frac{1}{ab^3},\frac{a}{b^3}).
$$
In order to study some useful features of the function $g$, let us
give the following lemma.
\begin{lem}\label{repelling-point}
1) Let $b> 1$, then the sequence $x_n =g(x_{n-1})$ converges to
$x_0=\sqrt[3]{\frac{u_1^{(0)}}{u_4^{(0)}}}=\alpha ^{2}>0$  under a
positive boundary condition, where $g$ is defined in
\eqref{c=1-fixed points12}.\\
2) Let $b<1$. then the sequence $y_n =h(y_{n-1})$ converges to
the initial point $y_0=\frac{\alpha ^2 \left(1+b^2 \alpha
^{6}\right)}{b^2+\alpha ^{6}}>0$ under a positive boundary
condition, where $h(x) = g(g(x)).$
\end{lem}
The proof of lemma \ref{repelling-point} is similar to \cite{RAU}.

\section{The phase diagrams of the model} 

In this section, we
 plot the phase diagrams associated with our model in in
three-parameter spaces. From the system \eqref{rec-eq2}, the total
partition function is given in terms of $({{u}_{i}})$ by
\[
Z^{(n)}=(u_{1}+u_{4})^{3}+(u_{5}+u_{8})^{3}.
\]
To discuss the phase diagram, we are interested in the
following reduced variables:
\begin{eqnarray}\label{eq8}
x^{(n)}=\frac{u_{4}^{(n)}+u_{5}^{(n)}}{u_1^{(n)}+u_{8}^{(n)}}, \
y^{(n)}=\frac{u_1^{(n)}-u_{8}^{(n)}}{u_1^{(n)}+u_{8}^{(n)}}, \
z^{(n)}=\frac{u_{4}^{(n)}-u_{5}^{(n)}}{u_1^{(n)}+u_{8}^{(n)}}.
\end{eqnarray}
The variable $x^{(n)}$ is just a measure of the frustration of the
nearest-neighbor bonds, and it is not an order parameter like
$y^{(n)},z^{(n)}$.

By using the equations \eqref{eq8}, we obtain the following
equalities for our computation purpose in the next equations as
$$
\frac{1+y^{(n)}}{2}=\frac{u_1^{(n)}}{u_1^{(n)}+u_8^{(n)}},\frac{1-y^{(n)}}{2}=\frac{u_8^{(n)}}{u_1^{(n)}+u_8^{(n)}},
$$
$$
 \frac{x^{(n)}+z^{(n)}}{2}=\frac{u_4^{(n)}}{u_1^{(n)}+u_8^{(n)}},\frac{x^{(n)}-z^{(n)}}{2}=\frac{u_5^{(n)}}{u_1^{(n)}+u_8^{(n)}}.
$$

From the equations \eqref{eq8}  and from the last equations, after some
lengthy and difficult calculations, we get the following
recurrence dynamical system:
\begin{eqnarray}\label{dynamical system2}
\left\{
\begin{array}{l}
 x^{(n+1)}=\frac{\left(c^2\left(1-y^{(n)}\right)+b^2\left(x^{(n)}-z^{(n)}\right)\right)^3+\left(1+y^{(n)}+b^2c^2\left(x^{(n)}+z^{(n)}\right)\right)^3}
{a^2\left(\left(x^{(n)}+b^2c^2\left(1+y^{(n)}\right)+z^{(n)}\right)^3+\left(b^2\left(1-y^{(n)}\right)+c^2\left(x^{(n)}-z^{(n)}\right)\right)^3\right)}\\ \\
 y^{(n+1)}=\frac{\left(x^{(n)}+b^2c^2\left(1+y^{(n)}\right)+z^{(n)}\right)^3-\left(b^2\left(1-y^{(n)}\right)+c^2\left(x^{(n)}-z^{(n)}\right)\right)^3}
{\left(x^{(n)}+b^2c^2\left(1+y^{(n)}\right)+z^{(n)}\right)^3+\left(b^2\left(1-y^{(n)}\right)+c^2\left(x^{(n)}-z^{(n)}\right)\right)^3}\\\\
z^{(n+1)}=\frac{\left(c^2\left(1-y^{(n)}\right)+b^2\left(x^{(n)}-z^{(n)}\right)\right)^3-\left(1+y^{(n)}+b^2c^2\left(x^{(n)}+z^{(n)}\right)\right)^3}
{a^2\left(\left(x^{(n)}+b^2c^2\left(1+y^{(n)}\right)+z^{(n)}\right)^3+\left(b^2\left(1-y^{(n)}\right)+c^2\left(x^{(n)}-z^{(n)}\right)\right)^3\right)}\,,
\end{array}
\right.
\end{eqnarray}
we assume $a=\exp(\alpha^{-1}), b=\exp(-\alpha^{-1}\beta)$ and
$c=\exp(-\alpha^{-1}\gamma)$, where $T/J=\alpha$, $-J_p/J=\beta$,
$-J_{t}/J=\gamma$, respectively.

The system obtained in \eqref{dynamical system2} is more
complicated than one might have anticipated. Therefore, it remains
difficult to be tackled analytically apart from simple limits and
numerical methods that are necessary to study the behavior of the
system in detail (see \cite{Vannimenus}).

Taking into account the positive boundary condition
$\bar{\sigma}^{(n)}(V\setminus V_n)\equiv 1$, we obtain initial
conditions
 \begin{eqnarray}\label{initial-con11}
\left\{
\begin{array}{l}
 x^{(1)}=\frac{1}{a^2c^6}, \\
 y^{(1)}=\frac{b^6-1}{b^6+1}, \\
 z^{(1)}=\frac{b^6-1}{a^2c^6(b^6+1)}.
\end{array}
\right.
\end{eqnarray}
For the fixed points, we consider the corresponding magnetization
$m$ given by
\begin{equation}\label{mag1}
m^{(n)}=\frac{\left(1+x^{(n)}+y^{(n)}+z^{(n)}\right)^3-\left(1+x^{(n)}-y^{(n)}-z^{(n)}\right)^3}
{\left(1+x^{(n)}+y^{(n)}+z^{(n)}\right)^3+\left(1+x^{(n)}-y^{(n)}-z^{(n)}\right)^3}.
\end{equation}
The initial point of the magnetization $m$ can be obtained as;
\begin{equation}\label{initial-mag1}
m^{(1)}=\frac{(b^2-1)\left((1+b^2)^2-b^2\right)\left((1+b^6)^2-b^6\right)}{\left(1+b^2\right)\left(1-b^2+b^4\right)\left(1-b^6+b^{12}\right)}.
\end{equation}
In order to plot the phase diagrams, we iterate the recurrence
equations \eqref{dynamical system2} and \eqref{mag1} and observe the
behavior of the phase diagrams after a large number of iterations
($n=10 000$).

A measure of the frustration of the nearest-neighbor bonds is
determined by the variable $x$ \cite{Vannimenus}. Furthermore, if $u_1 =
u_8$ and $u_4= u_5$, that is $y^{(n)}= z^{(n)}\rightarrow 0$, then
paramagnetic phase is obtained. 
If  $y^{(n)}\nrightarrow 0$ or $z^{(n)}\nrightarrow 0$, then one
obtains the ferromagnetic phase for the model (see
\cite{Vannimenus,MTA1985a} for details).

In the simplest situation, a fixed point $\textbf{u}^{*} =
(u^{*}_1, u^{*}_4, u^{*}_5, u^{*}_8)\in \mathbf{R}^{4}_{+}$ is
reached. Possible initial conditions with respect to  different
boundary conditions can be obtained in \cite{Vannimenus,MTA1985a}.
Here, we consider initial conditions \eqref{initial-con11} and
\eqref{initial-mag1}. Depending on $u^{*}_1, u^{*}_4, u^{*}_5,
u^{*}_8$, in the simplest situation a fixed point $(x^*,y^*,z^*)$
is reached. It corresponds to a {\it paramagnetic } phase (briefly
{\bf P}) if $y^*=0,z^*=0 $ or to a {\it ferromagnetic} phase
(briefly {\bf  F}) if $y^*,z^* \neq 0.$ From formula of average
magnetization \eqref{mag1} it follows that a situation where $y^*,z^*
\neq 0$ but $m=0$, cannot occur. Otherwise, the system has periodic
phases with the period $p$, i.e., the periodic phase is a configuration
with some period.

Let us consider the limit
$$
\lim_{n\rightarrow \infty }
(x^{(n)},y^{(n)},z^{(n)})=\lim_{n\rightarrow \infty }
(x^{(n+p)},y^{(n+p)},z^{(n+p)})=(x^{(*)},y^{(*)},z^{(*)}).
$$
The case $p=2$ corresponds to {\it antiferromagnetic} phase
(briefly {\bf P2}) and the case $p=4$ corresponds to the so-called
{\it antiphase} (briefly {\bf  P4}), denoted by \ $<2>$ for
compactness in \cite{Vannimenus,MTA1985a}. Finally, the system has
aperiodic phases, that is $p>11$, the system has modulated phases
(see \cite{UGAT2012IJMPC,MTA1985a,Inawashiro} for details). We
just consider periodic phases with period $p$ where $p\leqslant 12$
(briefly {\bf P2-P12}).
From the Lemma \ref{Lemma-fixed-set}, we can obtain the following
results:
\begin{itemize}
    \item a paramagnetic phase: if $\textbf{u}^{*}\in M_1$, in  figure \ref{fazk=3} (b), the white regions represent the paramagnetic
    phase. This represents  the set $M_1$ given in Lemma \ref{Lemma-fixed-set};
    \item a ferromagnetic phase: if $\textbf{u}^{*}\in M_2$, in  figure \ref{fazk=3} (b), the red regions represent the ferromagnetic
    phase. This represents  the set $M_2$ given in Lemma \ref{Lemma-fixed-set};
    \item in  figure \ref{fazk=3} (b), the yellow regions represent the
    \textbf{P2} phase. This represents  the set $M_3$ given in
    \eqref{2-period-set}.
\end{itemize}

\begin{figure}[!t]\centering
\includegraphics[width=55mm]{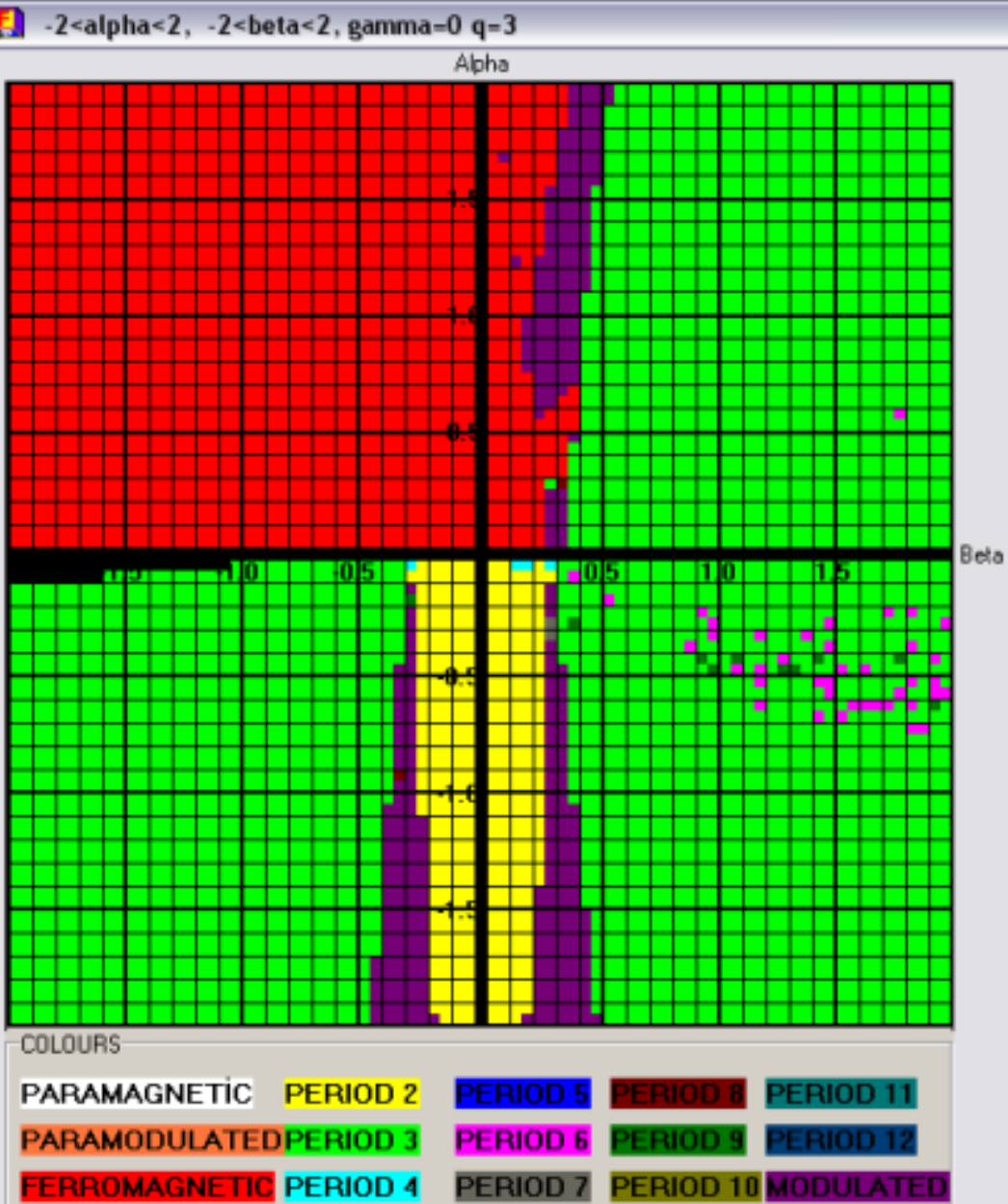}\ \ \ \ \ \ \ \ \ \ \ \
\includegraphics[width=55mm]{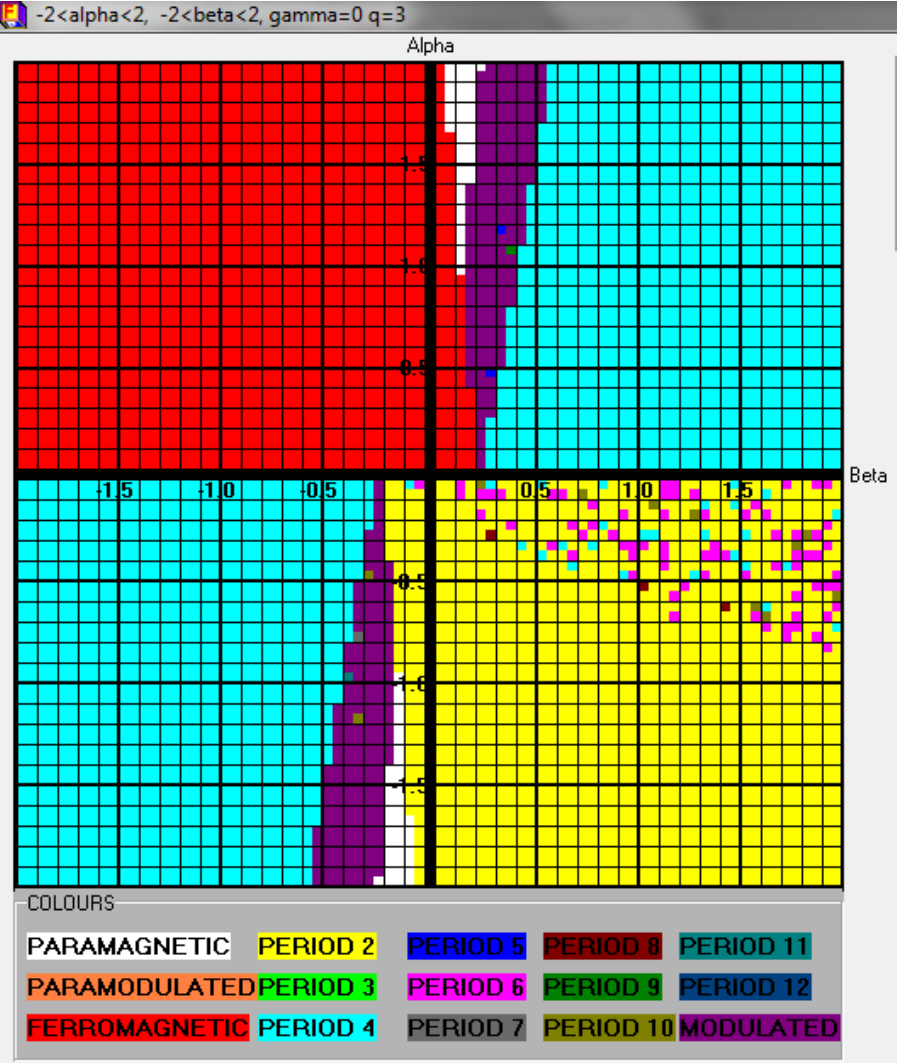}
\caption{(Colour online) (a) Phase diagram of the model for $J_p=0$
(the left-hand figure); (b) Phase diagram of the model for $J_t=0$ (the
right-hand figure).}\label{fazk=3}
\end{figure}

In figures \ref{fazk=3} (a) and (b), we  plotted the phase
diagrams of the model on the third-order Cayley tree \textbf{on the
rectangular region $[-2,2]\times [-2,2]$}.

In figure \ref{fazk=3} (a), we observe that the phase diagram
contains ferromagnetic (F), period-2, period-3 and modulated
phases, though the paramagnetic phase (white region) completely
disappears.

In the resultant phase diagram plotted by Vannimenus
\cite{Vannimenus}, the multicritical Libschit point is located in
$(-J_p/J,T/J)=(1/3,0)$, that is, the phases defined as the
ferromagnetic, paramagnetic, modulated and $<2>$ meet at the point
$(-J_p/J,T/J)=(1/3,0)$. In other words, all four phases meet at
the multicritical point $(p= 1/3,T=0)$. Contrary to Vannimenus's
work \cite{Vannimenus}, here the multicritical Libschit points
appear in non-zero points. In the resultant phase diagram in
figure \ref{fazk=3} (b), the ferromagnetic, paramagnetic,
modulated and $<2>$ phases intersect at the point
$(-J_p/J,T/J)=(0.15,0.39)$. Moreover,  the phases of ferromagnetic,
paramagnetic, modulated and \textbf{P}2 intersect at the
multicritical point $(J_p/J,-T/J)=(-0.15,-0.39)$.

In figure \ref{fazk=3} (b), the phase diagram consists of
ferromagnetic (F), paramagnetic (P) (fixed point), chaotic (C) (or
modulated), and antiferromagnetic + + - - (four cycle
antiferromagnetic phase) phases.

In the regions we have called the modulated (\textbf{M}) phase,
narrow regions with periodic orbit representing commensurate
phases have been identified. It is really difficult to detect
those with a periodicity higher than 12 in the modulated phase
regions. In order to determine the types of these phases, Lyapunov
exponent and the attractors are analyzed for some critical
parameters (see
\cite{UGAT2012IJMPC,MTA1985a,Inawashiro,Inawashiro-T1983}). The
modulated phases may generally contain two phases that correspond
to the devil's staircase \cite{MTA1985a}. These are called
commensurate (periodic) and incommensurate (aperiodic) regions,
respectively.

In Proposition \ref{proposition1}, we  analytically showed that
there exits a paramagnetic phase if and only if $b\in
[0,1]\cap[\sqrt{2},\infty )$. Therefore, we can present the following
theorem.

\begin{thm}\label{theorem-p} The model \eqref{Ham-Ist} (with $x \geqslant 0, \alpha > 0, b > 0$) has
a unique paramagnetic phase (Gibbs measure) if $b<1$. Assume that
if $b > \sqrt{2}$, then  the model \eqref{Ham-Ist} has exactly two
paramagnetic phases if either $\eta_1(b)=\alpha^{-2}$ or
$\eta_2(b)=\alpha^{-2}$. If $\eta_1(b)<\alpha^{-2}<\eta_2(b)$, 
then the model \eqref{Ham-Ist} has exactly three paramagnetic
phases.
\end{thm}
The proof of theorem \ref{theorem-p} is obtained from the
proposition \eqref{proposition1}. Note that from the equation
\eqref{eq8}, the corresponding phase is paramagnetic, because we take
into account $u_1^{(n)}=u_8^{(n)}$ and $u_4^{(n)}=u_5^{(n)}$, that
is, $x^{(n)}\rightarrow x^{*}$, $y^{(n)}\rightarrow 0$ and
$z^{(n)}\rightarrow 0$ (see \cite{Vannimenus} for details).

\begin{rem} We conclude that the phase boundaries between the periodic phases
presented in the diagrams are only approximate, because they are based
on the iteration of recursion relations in equation \eqref{dynamical
system2}, rather than on the minimization of the free energy.
Therefore, it is possible to calculate the free energy via the
scheme proposed by Gujrati \cite{Gujrati1995}. Therefore, one can
consider employing that scheme in order to improve the precision
of the phase diagrams. To calculate the free energy will not be
considered here.
\end{rem}

\begin{rem} In order to distinguish these phases from each other, one needs to
analyse the modulated phase regions via Lyapunov exponent and the
attractors in detail (see
\cite{UGAT2012IJMPC,MTA1985a,Inawashiro,Inawashiro-T1983}). Here,
we do not give these details.
\end{rem}
\subsection{The fixed points of the operator $F$ for $b=1$}
In the equation \eqref{rec-eq3}, if we assume as $b=1$, then we
have
\begin{eqnarray}\label{rec-eqb=1}
\left\{
\begin{array}{l}
 v_1=\frac{\alpha}{c}\left(c^2v_1^{3}+v_4^{3}\right) \\
 v_4=\frac{1}{\alpha c}\left(v_5^{3}+c^2v_8^{3}\right) \\
 v_5=\frac{1}{\alpha c}\left(v_1^{3}+c^2v_4^{3}\right)\\
 v_8=\frac{\alpha}{c}\left(c^2v_5^{3}+v_8^{3}\right).
\end{array}
\right.
\end{eqnarray}
Now, we describe the positive fixed points of system
\eqref{rec-eqb=1}.
Let us consider the following set
\begin{eqnarray}\label{non-para}
B:=\{(v_1,v_4,v_5,v_8)\in \mathbf{R}^4_+:v_1=v_8,v_4=v_5\}.
\end{eqnarray}
\begin{rem}
From the system \eqref{rec-eqb=1}, it is clear that the equations
$v_1=v_8$ and $v_4=v_5$ do not satisfy, i.e., the set $B$ given in
\eqref{non-para} is empty. Therefore, the paramagnetic phase
regions in the phase diagrams associated with the model disappear
in $[-2,2]\times [-2,2]\subset \mathbf{R}^{2}$.
\end{rem}
We can conclude that there exists at least one fixed point of the
corresponding $g$ function if $u_1=u_8, u_4 = u_5$ and $c = 1$,
that is, the system has a parametric phase [see figure \ref{fazk=3}
(b)]. Conversely, if $b = 1$ then $u_1=u_8, u_4 = u_5$ are not
satisfied, so the corresponding operator has no fixed points. In
this case, the system has no paramagnetic phase. However, the same system
has a ferromagnetic phase [see figure~\ref{fazk=3}~(a)].

\section{Conclusions}\label{Conclusions}
Written for both mathematics and physics audience, this paper has
a fourfold purpose:
\begin{enumerate}
    \item[1)] to analytically study  the recurrence
equations associated with the model \eqref{Ham-Ist};
    \item[2)] to numerically obtain
 the paramagnetic, the ferromagnetic and period 2
regions corresponding to the sets $M_1,M_2, B$, respectively;
    \item[3)]
to illustrate the fixed points of the corresponding operator;
\item[4)] to compare the numerical results with the exact solutions of
the model.
\end{enumerate}
We state some unsolved problems that turned out to be rather
complicated and require a further consideration:
\begin{enumerate}
    \item Do any other invariant sets of the operator $F$ exist?
    \item Do positive fixed points of the operator $F$ exist outside the invariant sets?
    \item Does there exist a periodic point ($p>2$) of rather cumbersome high-order equations that can be solved using
      analytic methods?
\end{enumerate}
In the first case, we have already obtained the fixed points of the
operator $F$ so that $u_1=u_8$ and $u_4=u_5$. For the periodic
case, however, it is not possible to obtain all solutions
satisfying all requirements of equations \eqref{rec-eq2}, so that
$\{\textbf{u}=(u_1,u_4,u_5,u_8)\in
\mathbf{R}^4_+:F^p(\textbf{u})=\textbf{u},p>1\}$ is invariant. The
proof of this statement  involves  a number mathematical
complexities. Moreover, in the second case ($b=1$), it is much more difficult to analytically find 
the fixed points of the operator  $F$.

Vannimenus \cite{Vannimenus} showed that at $T=0$, only two
different ground-states are encountered: the ferromagnetic state,
for $p(=-J_p/J)$ smaller than 1/3, and a state of period 4 in
which the magnetization of the successive generations alternates
with an antiferromagnetic structure $(+ + - -)$, for $p > 1/3$. By
using the standard approach, we have analytically proved that
there exits a paramagnetic phase when $J_p> 0$ and for the phase with
period 2 when $J_p<0$. Our results do not contrast with those of
Vannimenus~\cite{Vannimenus}. In this present paper, we show that
a paramagnetic phase completely disappears at $T=0$.
We plan to experiment in our future work in order to analyse the same
problem analytically for the high-order Cayley tree. Recently, the
author \cite{AkinCMP2019} has studied the existence of the Gibbs
measures of an Ising model with competing interactions on the
triangular chandelier-lattice. The phase diagrams corresponding to
the Gibbs states on the Cayley-like lattices have not been examined
yet. In the future papers, we are planning to investigate the same
problems (see \cite{AkinRG2020}).
\section*{Acknowledgements}
The author thanks  the referee for the careful reading of the
manuscript and for the valuable comments and suggestions.

\ukrainianpart

\title{Визначення парамагнітної та феромагнітної фаз моделі Ізінга на дереві Кейлі третього порядку}
\author {Г. Акін}
\address{вул. Джейхун Атуф Кансу 1164, Сокак, 9/4, TR06105, Чанкая, Анкара, Туреччина}

\makeukrtitle 

\begin{abstract}
У статті отримано рекурентні рівняння  моделі Ізінга
з трьома константами зв'язку на дереві Кейлі третього порядку. Представлено характеристики парамагнітна та феромагнітної фази, пов'язані з
 моделлю Ізінга. Строго вивчаються типи фаз та статистична сума, що відповідають моделі. Точні
розв'язки згаданої моделі порівнюються з числовими
результатами, наведеними в  Ganikhodjaev {et al.} [ J. Concrete and
Applicable Mathematics, 2011, \textbf{9}, No.~1, 26--34].

\keywords  дерево Кейлі, модель Ізінга, парамагнітна фаза,
феромагнітна фаза 

\end{abstract}

\lastpage


\begin{thebibliography}{11}

\bibitem{Inawashiro} Inawashiro S., Thompson C.J.,  Honda G.,
J. Stat. Phys., 1983, \textbf{33}, 419--436,
\doi{10.1007/BF01009804}.

\bibitem{Vannimenus} Vannimenus J.,
Z. Phys. B: Condens. Matter, 1981, {\textbf 43}, No.~2, 141--148, \doi{10.1007/BF01293605}.

\bibitem{NHSS} Ganikhodjaev N.N.,   Ak\i n H.,  Uguz S., Temir T.,
J. Stat. Mech: Theory Exp., 2011, \textbf{03}, P03025, \\\doi{10.1088/1742-5468/2011/03/P03025}.

\bibitem{UGAT2012IJMPC}  Uguz S., Ganikhodjaev N.N.,  Ak\i n H.,  Temir S.,
Int. J. Mod. Phys. C, 2012, \textbf{23},\\ No.~5, 1250039, \doi{10.1142/S0129183112500398}.

\bibitem{MTA1985a} Mariz M., Tsalis C., Albuquerque A.L., J. Stat. Phys., 1985,
\textbf{40}, 577--592, \doi{10.1007/BF01017186}.

\bibitem{Yokoi-Oliveira-S1985}  Yokoi C.S.O.,  de Oliveira M.J., Salinas S.R., 
Phys. Rev. Lett., 1985, \textbf{54}, No.~3, 163--166,\\
\doi{10.1103/PhysRevLett.54.163}.

\bibitem{Lebowitz1} Lebowitz J.L., J. Stat. Phys., 1977, \textbf{16}, No.~6, 463--476, \doi{10.1007/BF01152284}.

\bibitem{GAUT2011Chaos} Ganikhodjaev N.N.,  Ak\i n H., Uguz S., Temir S.,
J. Concr. Appl. Math., 2011, \textbf{9}, No.~1, 26--34.

\bibitem{GU2011a} Ganikhodjaev N.N., Uguz S.,
Physica A, 2011, \textbf{390}, No.~23--24,
4160--4173, \doi{10.1016/j.physa.2011.06.044}.

\bibitem{GTA}  Ganikhodjaev N.N., Temir S., Ak\i n H.,
J. Stat. Phys., 2009, {\bf 137},
701--715, \doi{10.1007/s10955-009-9869-z}.

\bibitem{NHSS1} Ganikhodjaev N.N., Ak\i n H.,  Uguz S., Temir S.,
Phase Transitions, 2011, {\bf 84}, No.~11--12,
1045--1063, \\\doi{10.1080/01411594.2011.579395}.

\bibitem{AUT2010AIP}  Ak\i n H., Uguz S., Temir S., AIP Conf. Proc., 2010, \textbf{1281}, 607--611,
\doi{10.1063/1.3498550}.

\bibitem{UA2011CJP}  Uguz S., Ak\i n H.,
Chin. J. Phys., 2011, \textbf{49}, No.~3, 788--801.

\bibitem{UA2010PhysicaA} Uguz S., Ak\i n H.,
Physica A, 2010, \textbf{389}, 1839, \doi{10.1016/j.physa.2009.12.057}.

\bibitem{UGAT2012ACTA} Uguz S., Ganikhodjaev N.N.,  Ak\i n H., Temir S.,
Acta Phys. Pol. A, 2012, 
\textbf{121}, No.~1, 114--118,\\ \doi{10.12693/APhysPolA.121.114}.

\bibitem{Inawashiro-T1983}  Inawashiro S.,  Thompson C.J.,
Phys. Lett. A, 1983, \textbf{97}, 245--248, \doi{10.1016/0375-9601(83)90758-2}.

\bibitem{Moraal} Moraal  H.,
Physica A, 1978, \textbf{92}, 305--314, \doi{10.1016/0378-4371(78)90037-7}.

\bibitem{Chakraborty1992} Chakraborty K.G., J. Magn. Magn. Mater., 1992,
\textbf{114},  No.~1--2, 155--160, \doi{10.1016/0304-8853(92)90340-T}.

\bibitem{Horiguchi1986} Horiguchi T.,
Physica A, 1986, \textbf{136}, No.~1, 109--123,
\doi{10.1016/0378-4371(86)90045-2}.

\bibitem{Azhari2017}  Azhari M., Benayad N.,  Mouhib M., Phase Transitions, 2017,
\textbf{90}, No.~5, 485--499, \\\doi{10.1080/01411594.2016.1227985}.

\bibitem{Akin-Saygili2015} Ak\i n H., Sayg\i l\i \ H.,
AIP Conf. Proc., 2015,  \textbf{1676}, 020026, \doi{10.1063/1.4930452}.

\bibitem{ART}  Ak\i n H.,   Rozikov U.A., Temir S.,
J. Stat. Phys., 2011, \textbf{142}, No.~2, 314--321, \doi{10.1007/s10955-010-0106-6}.

\bibitem{Nazarov-Rozikov} Nazarov  Kh.A.,  Rozikov U.,
 Theor. Math. Phys., 2003, \textbf{135}, No.~3, 881--888, \doi{10.1023/A:1024091206594}.

\bibitem{AGUT}  Ak\i n H., Ganikhodjaev N.,  Uguz S., Temir S.,
AIP Conf. Proc., 2011, \textbf{1389}, No.~1, 2004--2007, \\\doi{10.1063/1.3637008}.

\bibitem{AGTU2013ACTA} Ak\i n H., Ganikhodjaev N.N., Temir S., Uguz S.,
Acta Phys. Pol. A, 2013, \textbf{123}, No.~2, 484--487, \\\doi{10.12693/APhysPolA.123.484}.

\bibitem{Akin2017} Ak\i n H.,
Int. J. Mod. Phys. B, 2017, \textbf{31}, No.~13, 1750093, \doi{10.1142/S021797921750093X}.

\bibitem{Akin2016}  Ak\i n H.,
Chin. J. Phys., 2016, {\bf 54}, No.~4,
635--649, \doi{10.1016/j.cjph.2016.07.010}.

\bibitem{Akin2017a} Ak\i n H.,
Int. J. Mod. Phys. C, 2018, \textbf{29}, No.~2, 1850016, \doi{10.1142/S012918311850016X}.

\bibitem{AT1} Ak\i n H., Temir S.,
Condens. Matter Phys, 2011, 
\textbf{14}, No.~2, 23003, \doi{10.5488/CMP.14.23003}.

\bibitem{BRZ}  Bleher P.M.,  Ruiz J., Zagrebnov V.A.,
J. Stat. Phys., 1995, \textbf{79}, 473--482,
\doi{10.1007/BF02179399}.

\bibitem{BleherG}  Bleher P.M., Ganikhodjaev N.N.,
Theory Probab. Appl., 1990, \textbf{35}, 216--227,  \doi{10.1137/1135031}.

\bibitem{RAU}  Rozikov U.A.,  Ak\i n H.,  Uguz S.,
Math. Phys. Anal. Geom., 2014, {\bf 17}, 103--114, \doi{10.1007/s11040-014-9144-7}.

\bibitem{Rozikov} Rozikov U.A., Gibbs Measures on Cayley Trees, World Scientific, Singapore, 2013.

\bibitem{Bleher-Zalys} Bleher P., Zalys E., Lith. Math. J., 1988, \textbf{28}, No.~2, 127.

\bibitem{Bleher1990a} Bleher P.M.,
Commun. Math. Phys., 1990,
\textbf{128}, No.~2, 411--419, \doi{10.1007/BF02108787}.

\bibitem{GATTMJ}  Ganikhodjaev N.N.,  Ak\i n H., Temir T.,
Turk. J.
Math., 2007, {\bf 31}, No.~3, 229--238.

\bibitem{Gandolfo2012}   Gandolfo D.,  Ruiz J.,    Shlosman~S.,
 J. Stat. Phys., 2012, \textbf{148}, 999--1005,
\doi{10.1007/s10955-012-0574-y}.

\bibitem{Kindermann} Kindermann R., Snell J.L., Markov Random Fields and Their Applications, Amer. Math. Soc., Providence, Rhode Island, 1980.

\bibitem{MAKfree2017} Mukhamedov F., Ak\i n H., Khakimov O.,
J. Stat. Mech., 2017, \textbf{2017}, 053208,
\doi{10.1088/1742-5468/aa6c88}.

\bibitem{AkinRG2020}  Ak\i n H., 
(unpublished),
URL~\url{https://www.researchgate.net/publication/340540145}.

\bibitem{Preston} Preston Ch.J., Gibbs States on Countable Sets, Cambridge Univ. Press, Cambridge, 1974.

\bibitem{GTA-CUBO2005} Ganikhodjaev N., Temir S., Ak\i n H., CUBO, A Mathematical J., 2005, \textbf{7}, No.~3, 39--48.

\bibitem{GR-2019} Ganikhodjaev N., Rahmatullaev A.,
 Phase Transitions, 2019,
\textbf{92}, No.~8, 730--736, \\\doi{10.1080/01411594.2019.1639700}.

\bibitem{Wolfram} Wolfram Research, Inc., Mathematica, Version 8.0, Champaign, IL, 2010.

\bibitem{GR}  Ganikhodjaev N.N., Rozikov U.A.,
Theor. Math. Phys., 1997, {\bf 111},
480--486, \doi{10.1007/BF02634202}.


\bibitem{Gujrati1995} Gujrati P.D., Phys. Rev.
Lett., 1995, \textbf{74}, 809, \doi{10.1103/PhysRevLett.74.809}.


\bibitem{AkinCMP2019}  Ak\i n H., Condens.
Matter Phys., 2019, \textbf{22}, No.~2, 23002, 1--14, \doi{10.5488/CMP.22.23002}.

\end{thebibliography}
\end{document}